%% file: wc-arXiv.tex
\documentclass{endm}
\usepackage{endmmacro}

\usepackage{xspace}

\input{Package}
\input{Command}

\begin{document}

\begin{verbatim}\end{verbatim}\vspace{2.5cm}

\begin{frontmatter}
\title{Ruling out FPT algorithms for\\Weighted Coloring on forests\thanksref{ALL}}

\author{J\'ulio Ara\'ujo$^\dag$, Julien Baste$^\ddag$, and Ignasi Sau$^{\dag, \ddag}$}
\address{$^\dag$ Departamento de Matem\'atica, Universidade Federal do Cear\'a, Fortaleza, Brazil\\
\emph{Email: \href{mailto:julio@mat.ufc.br} {\texttt{\normalshape  julio@mat.ufc.br}}}\\
$^\ddag$ CNRS, LIRMM, Universit\'e de Montpellier, Montpellier, France\\
\emph{Emails: \href{mailto:sau@lirmm.fr} {\texttt{\normalshape  \{baste,sau\}@lirmm.fr}}}}

%

\thanks[ALL]{This work has been partially supported by CNPq/Brazil under projects 459466/2014-3 and 310234/2015-8, and by the PASTA project of Universit\'e de Montpellier, France.}


\begin{abstract}
Given a graph $G$, a \emph{proper $k$-coloring} of $G$ is a partition $c = (S_i)_{i\in [1,k]}$ of $V(G)$ into $k$ stable sets $S_1,\ldots, S_{k}$. Given a weight function $w: V(G) \to \Rbb^+$, the \emph{weight of a color} $S_i$ is defined as $w(i) = \max_{v \in S_i} w(v)$ and the \emph{weight of a coloring} $c$  as $w(c) = \sum_{i=1}^{k}w(i)$.  Guan and Zhu [Inf. Process. Lett., 1997] defined the \emph{weighted chromatic number} of a pair $(G,w)$, denoted by $\sigma(G,w)$, as the minimum weight of a proper coloring of $G$. For a positive integer $r$, they also defined $\sigma(G,w;r)$ as the minimum of $w(c)$ among all proper $r$-colorings $c$ of $G$.

The complexity of determining $\sigma(G,w)$ when $G$ is a tree was  open for almost 20 years, until Ara\'ujo \emph{et al}.~[SIAM J. Discrete Math., 2014] recently proved that the problem cannot be solved in time $n^{o(\log n)}$ on $n$-vertex trees unless the Exponential Time Hypothesis (ETH) fails.


The objective of this article is to provide hardness results for computing $\sigma(G,w)$ and $\sigma(G,w;r)$ when $G$ is a tree or a forest, relying on complexity assumptions weaker than the ETH. Namely, we study the problem from the viewpoint of parameterized complexity, and we assume the weaker hypothesis $\fpt \neq \W[1]$. Building on the techniques of Ara\'ujo \emph{et al}., we prove that when $G$ is a forest, computing $\sigma(G,w)$ is {\sf W}[1]-hard parameterized by the size of a largest connected component of $G$, and that computing $\sigma(G,w;r)$ is {\sf W}[2]-hard parameterized by $r$.
Our results rule out the existence of \fpt algorithms for computing these invariants on trees or forests for many natural choices of the parameter.


\end{abstract}


\end{frontmatter}

\section{Introduction}
\label{sec:intro}

A \emph{(vertex) $k$-coloring} of a graph $G = (V,E)$ is a function $c:V(G)\to \{1,\ldots, k\}$. Such coloring $c$ is \emph{proper} if $c(u)\neq c(v)$ for every edge $\{u,v\} \in E(G)$. All the colorings we consider in this paper are proper, hence we may omit the word ``proper''. The \emph{chromatic number} $\chi(G)$ of $G$ is the minimum integer $k$ such that $G$ admits a $k$-coloring. Given a graph $G$, determining $\chi(G)$ is the goal of the classical \textsc{Vertex Coloring} problem.
If $c$ is a $k$-coloring of $G$, then $S_i = \{u\in V(G)\mid c(u) = i\}$ is a stable (a.k.a. independent) set. Consequently, a $k$-coloring $c$ can be seen as a partition of $V(G)$ into stable sets $S_1,\ldots, S_k$. We often see a coloring as a partition in the sequel.

We study a generalization of \textsc{Vertex Coloring} for vertex-weighted graphs that has been defined by Guan and Zhu~\cite{GZ97}. Given a graph $G$ and a weight function $w: V(G) \to \Rbb^+$, the \emph{weight of a color} $S_i$ is defined as $w(i) = \max_{v \in S_i} w(v)$. Then, the \emph{weight of a coloring} $c$ is $w(c) = \sum_{i=0}^{k-1}w(i)$. In the \textsc{Weighted Coloring} problem, the goal is to determine the \emph{weighted chromatic number} of a pair $(G,w)$, denoted by $\sigma(G,w)$, which is the minimum weight of a coloring of $(G,w)$. A coloring $c$ of $G$ such that $w(c) = \sigma(G,w)$ is an \emph{optimal weighted coloring}. Guan and Zhu~\cite{GZ97} also defined, for a positive integer $r$, $\sigma(G,w;r)$ as the minimum of $w(c)$ among all $r$-colorings $c$ of $G$, or as $+\infty$ is no $r$-coloring exists. Note that $\sigma(G,w) = \min_{r \geq 1}\sigma(G,w;r)$.  It is worth mentioning that the \textsc{Weighted Coloring} problem is also sometimes called \textsc{Max-Coloring} in the literature; see for instance~\cite{KaMe09,PemmarajuPR05}.

Guan and Zhu defined this problem in order to study practical applications related to resource allocation, which they describe in detail in~\cite{GZ97}.
One should observe that if all the vertex weights are equal to one, then $\sigma(G,w) = \chi(G)$, for every graph $G$. Consequently, determining $\sigma(G,w)$ and $\sigma(G,w;r)$ are {\sf NP}-hard problems on general
graphs~\cite{Karp72}. In fact, this problem has been shown to be {\sf NP}-hard even on split graphs, interval graphs, triangle-free planar graphs with bounded degree, and bipartite graphs~\cite{DWMP02, WDEMP09, EMP06}. On the other hand, the weighted chromatic
number of cographs and of some subclasses of bipartite graphs can be found in
polynomial time~\cite{DWMP02, WDEMP09}.

In this work we focus on the case where $G$ is a forest, which has attracted considerable attention in the literature. Concerning graphs of bounded treewidth\footnote{We will not define treewidth here, just recall that forests are the graphs with treewidth 1; see~\cite{Die05,CyganFKLMPPS15}.},  Guan and Zhu~\cite{GZ97} showed, by using standard dynamic programming techniques, that
on an $n$-vertex graph of treewidth $t$ the parameter $\sigma(G,w;r)$ can be computed in time
\begin{equation}\label{eq:algorithm}
n^{O(r)} \cdot r^{O(t)}.
\end{equation}
Guan and Zhu~\cite{GZ97} left as an open problem whether \textsc{Weighted Coloring} is polynomial on trees and, more generally, on graphs of bounded treewidth. Escoffier \emph{et al}.~\cite{EMP06} found a polynomial-time approximation scheme to solve \textsc{Weighted Coloring} on bounded treewidth graphs, and
Kavitha and Mestre~\cite{KaMe09} showed that the problem is in \poly on the class of trees where vertices with degree at
least three induce a stable set.

But the question of Guan and Zhu has been answered only recently, when  Ara\'ujo \emph{et al}.~\cite{Julio} showed that, unless the Exponential Time Hypothesis (ETH)\footnote{The ETH states that 3-\textsc{SAT} cannot be solved in subexponential
time; see~\cite{ImpagliazzoP01} for more details.} fails, there is no algorithm computing the weighted chromatic number of $n$-vertex trees
in time $n^{o(\log n)}$.

As discussed in~\cite{Julio}, it is worth mentioning that the above lower bound is tight. Indeed, Guan and Zhu~\cite{GZ97} showed that the maximum number of colors used by an optimal weighted coloring of any weighted graph $(G,w)$ is at most its so-called \emph{first-first chromatic number} (see~\cite{GZ97} for the definition), denoted by $\xff(G)$. On the other hand, Linhares and Reed~\cite{SaRe06} proved that for any $n$-vertex graph $G$ of treewidth at most $t$, it holds that  $\xff(G) = O(t \log n)$. These observations together with  Equation~(\ref{eq:algorithm}) imply that the \textsc{Weighted Coloring} problem can be solved on forests in time $n^{O(\log n)}$.

Therefore, \textsc{Weighted Coloring} on forests is unlikely to be in \poly, as this would contradict the ETH, and also unlikely to be \np-hard, as in that case all problems in \np could be solved in subexponential time, contradicting again the ETH.

\smallskip
\noindent \textbf{Our results}. The objective of this article is to provide hardness results for computing $\sigma(G,w)$ and $\sigma(G,w;r)$ when $G$ is a forest, relying on complexity assumptions weaker than the ETH. Namely, we study the problem from the viewpoint of parameterized complexity (the basic definitions can be found in Section~\ref{sec:prelim}), and we assume the weaker hypothesis $\fpt \neq \W[1]$. Indeed, it is well-known~\cite{CyganFKLMPPS15} that the ETH implies that $\fpt \neq \W[1]$, which in turn implies that ${\sf P} \neq {\sf NP}$.

Our first result is that when $(G,w)$ is a weighted forest, computing $\sigma(G,w)$ is {\sf W}[1]-hard parameterized by the size of a largest connected component of $G$. This is proved by a parameterized reduction from \textsc{Independent Set} that builds on the techniques introduced by Ara\'ujo \emph{et al}.~\cite{Julio}. This result essentially rules out the existence of \fpt algorithms for \textsc{Weighted Coloring} on forests for many natural choices of the parameter: treewidth, cliquewidth,  maximum degree, maximum diameter of a connected component, number of colors in an optimal coloring, etc. Indeed, all these parameters are bounded by the size of a largest connected component of $G$ (for the number of colors, this can be proved by using that they are bounded by $\xff(G)$~\cite{GZ97}, which is easily seen to be bounded by the size of a largest connected component).

We then move our attention to the parameter $\sigma(G,w;r)$ and we prove, by a parameterized reduction from \textsc{Dominating Set} similar to the first one, that computing $\sigma(G,w;r)$ on forests is {\sf W}[2]-hard parameterized by $r$. Interestingly, if we assume the ETH, our reduction together with the results of Chen \emph{et al}.~\cite{ChenHKX06} imply that, on graphs of bounded treewidth, the running time given by Equation~(\ref{eq:algorithm}) is asymptotically optimal, that is, there is no algorithm computing  $\sigma(G,w;r)$ on $n$-vertex graphs of bounded treewidth in time $n^{o(r)}$.

We would like to mention that, although our reductions use several key ideas introduced by Ara\'ujo \emph{et al}.~\cite{Julio}, our results are incomparable to those of~\cite{Julio}.

As further research, it would be interesting to identify ``reasonable'' parameters yielding \fpt algorithms for \textsc{Weighted Coloring} on forests. Probably, it might make sense to consider combined parameters that take into account, on top of  the aforementioned invariants, the number of distinct weights in the input weighted forest.

\medskip
\noindent
\textbf{Organization of the article.} In Section~\ref{sec:prelim} we provide some basic preliminaries about forests, weighted colorings, and parameterized complexity. In Section~\ref{sec:gadgets} we introduce some common gadgets that will be used in both reductions. In Section~\ref{sec:W1} and Section~\ref{sec:W2} we present the $\W[1]$-hardness and $\W[2]$-hardness reductions, respectively.

\section{Preliminaries}
\label{sec:prelim}

\textbf{Forests and weighted colorings}. We use standard graph-theoretic notation, and we consider simple undirected graphs without loops nor multiple edges; see~\cite{Die05} for any undefined terminology. Given two integers $i$ and $j$ with $i \leq j$, we denote by $\intv{i,j}$ the set of all integers between $i$ and $j$, including both $i$ and $j$.

If $T$ is a rooted tree, we denote by $r(T)$ the root of $T$.
A \emph{weighted graph} is a pair $(G,w)$ where $G$ is a graph and  $w: V(G) \to \Rbb^+$ is a weight function.
We say that a weighted graph $(G,w)$ is a \emph{weighted forest} if $G$ is a forest and a \emph{weighted rooted tree} if $G$ is a rooted tree.
If $(G,w)$ is a weighted rooted tree, we define the \emph{root} of $(G,w)$, denoted by $r((G,w))$, to be the root of $G$.

When considering $k$-colorings $c$ of a graph $G$, defined in Section~\ref{sec:intro}, for convenience we will usually index them as $c = (S_i)_{i \in \intv{0,k-1}}$. We say that a vertex $v \in V(G)$ is \emph{colored} $S_i$, for some $i \in \intv{0,k-1}$, if $v \in S_i$.

\medskip
\noindent \textbf{Parameterized complexity.} We refer the reader to~\cite{DF13,CyganFKLMPPS15} for basic background on parameterized complexity, and we recall here only some basic definitions.
A \emph{parameterized problem} is a language $L \subseteq \Sigma^* \times \mathbb{N}$.  For an instance $I=(x,k) \in \Sigma^* \times \mathbb{N}$, $k$ is called the \emph{parameter}. 
A parameterized problem is \emph{fixed-parameter tractable} ({\sf FPT}) if there exists an algorithm $\Acal$, a computable function $f$, and a constant $c$ such that given an instance $I=(x,k)$,
$\Acal$ (called an {\sf FPT} \emph{algorithm}) correctly decides whether $I \in L$ in time bounded by $f(k) \cdot |I|^c$.


%

Within parameterized problems, the class {\sf W}[1] may be seen as the parameterized equivalent to the class \np of classical optimization problems. Without entering into details (see~\cite{DF13,CyganFKLMPPS15} for the formal definitions), a parameterized problem being {\sf W}[1]-\emph{hard} can be seen as a strong evidence that this problem is {\sl not} \fpt. The canonical example of {\sf W}[1]-hard problem is \textsc{Independent Set} parameterized by the size of the solution\footnote{Given a graph $G$ and a parameter $k$, the problem is to decide whether there exists $S \subseteq V(G)$ such that $|S| \geq k$ and $E(G[S]) = \es$.}.

The class {\sf W}[2] of parameterized problems is a class that contains $\W$[1], and such that the problems that are {\sf W}[2]-\emph{hard} are  even more unlikely to be \fpt than those that are {\sf W}[1]-hard (again, see~\cite{DF13,CyganFKLMPPS15} for the formal definitions). The canonical example of {\sf W}[2]-hard problem is \textsc{Dominating Set} parameterized by the size of the solution\footnote{Given a graph $G$ and a parameter $k$, the problem is to decide whether there exists $S \subseteq V(G)$  such that  $|S| \leq k$ and $N[S] = V(G)$.}.

For $i \in \intv{1,2}$, to transfer ${\sf W}[i]$-hardness from one problem to another, one uses a \emph{parameterized reduction}, which given an input $I=(x,k)$ of the source problem, computes in time $f(k) \cdot |I|^c$, for some computable function $f$ and a constant $c$, an equivalent instance $I'=(x',k')$ of the target problem, such that $k'$ is bounded by a function depending only on $k$.

Hence, an equivalent definition of $\W$[1]-hard (resp. $\W$[2]-hard) problem is any problem that admits a parameterized reduction from \textsc{Independent Set} (resp. \textsc{Dominating Set}) parameterized by the size of the solution.

\section{Some useful gadgets}
\label{sec:gadgets}


In this section we introduce some gadgets that will be used in the reductions presented in the following sections. As mentioned in the introduction, the first reduction is from \textsc{Independent Set}, and the second one is from \textsc{Dominating Set}. Most of these gadgets are inspired by~\cite{Julio}.


Let us first fix $(G,k)$, an instance of either \textsc{Independent Set} or \textsc{Dominating Set}.
We denote by $(G',w)$ the instance of  \textsc{Weighted Coloring} we are going to construct.
We define $n = |V(G)|$ and we fix a bijection $\beta: V(G) \to \intv{0,n-1}$, which will allow us to define our gadgets depending on integers $j \in \intv{0,n-1}$ corresponding, via $\beta$,  to the vertices of $G$.
We define  $M = k(n-1)\varepsilon + \sum_{i\in \intv{0,4k+3}}\frac{1}{2^i}$, where $\varepsilon$ is any real number satisfying $0 < \varepsilon < \frac{1}{n k 2^{4k+3}}$, which implies that $M < 2$.
We define, for each $i \in \intv{0,4k+3}$ and for each $j \in \intv{0,n}$, $w_i^j = \frac{1}{2^i}+j\varepsilon$.
We also define, for each $\ell \in \intv{0,3}$, $W_\ell = w_{4n+\ell}^0 = \frac{1}{2^{4n+\ell}}$.

We first define a particular family of \emph{binomial trees} $B_i$, $i\in \intv{0,4n + 3}$, depicted in Figure~\ref{fig:bin}, as done in \cite{Julio}.
They will be crucially used in the construction of $(G',w)$.
Their role is to force the color of most of the nodes in any coloring $c$ of $G'$ with $w(c) ≤ M$.
Note that the notion of binomial trees has also been used, for instance, in \cite{WDEMP09,BonnetF0S15}.

\begin{definition}
  For each $i \in \intv{0,4k+3}$, we define recursively the weighted rooted tree $B_i$, called \emph{binomial tree}, as follows:
  \begin{itemize}
  \item if $i = 0$, then $B_0$ has a unique node of weight $w_0^0$,
  \item otherwise, $B_i$ has a root $r$ of weight $w_i^0$ and, for each $j \in \intv{0,i-1}$, we introduce a copy of $B_j$ and we connect its root to $r$.
  \end{itemize}
\end{definition}

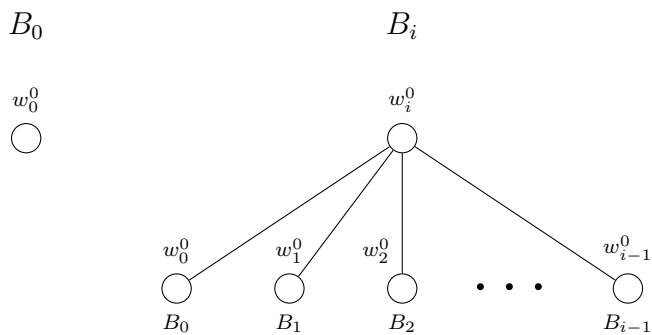
\begin{figure}[h!] 
  \centering
  \begin{tikzpicture}
    \node[draw, circle, black, label={90:\scriptsize{$w_0^0$}}]  at (0,0){};
    \node[black] (w00) at (0,1.5){$B_0$};

    \node[draw, circle, black, label={90:\scriptsize{$w_i^0$}}] (wi0) at (5,0){};
    \node[black] (w00) at (5,1.5){$B_i$};

    \node[draw, circle, black, label={90:\scriptsize{$w_0^0$}}, label={-90:\scriptsize{$B_0$}}] (w00) at (2,-2){};
    \node[draw, circle, black, label={90:\scriptsize{$w_1^0$}}, label={-90:\scriptsize{$B_1$}}] (w10) at (3.5,-2){};
    \node[draw, circle, black, label={95:\scriptsize{$w_2^0$}}, label={-90:\scriptsize{$B_2$}}] (w20) at (5,-2){};
    \node[black] at (6.5,-2){\Huge{$\cdots$}};
    \node[draw, circle, black, label={90:\scriptsize{$w_{i-1}^0$}}, label={-90:\scriptsize{$B_{i-1}$}}] (wi10) at (8,-2){};

    \draw (wi0) -- (w00);
    \draw (wi0) -- (w10);
    \draw (wi0) -- (w20);
    \draw (wi0) -- (wi10);

  \end{tikzpicture}
  \caption{
The binomial trees $B_0$ and $B_i$, $i > 0$.
The vertices labeled $B_j$ are the root of a copy of $B_j$, for each $j \in \intv{0,i-1}$. The weights are also depicted on top of the vertices.
}

  \label{fig:bin}
\end{figure}

\begin{lemma}[Ara\'ujo \emph{et al.}\cite{Julio}]
\label{lemma:bin}
Let $i \in \intv{0,4k+3}$ and let $(T,w)$ be a weighted forest having $B_i$ as
a subtree. If there exists a coloring $c$ of $(T, w)$ with $w(c) ≤ M$, then, for any $\ell \in \intv{0,i}$:
\begin{itemize}
\item all vertices of $B_i$ with weight in $w_\ell^0$ receive the same color $S_\ell$ of $c$ and
\item there exists a unique color class $S_\ell$ in $c$ of  weight in $\{w_\ell^j \mid j \in \intv{0,n}\}$.
\end{itemize}
\end{lemma}

As we shall see later, the choice of the weight of a color class $S_\ell$ corresponds to choosing (or not) a vertex to be part of the solution of the corresponding problem. Each time that a vertex is chosen, we will have to ``pay'' an additional weight of $(n-1)\varepsilon$ in the total weight of the coloring. The selected value of $M$ forces that we will be able to choose $k$ vertices.

In every graph we are going to build in the following, we assume that $B_{4k+3}$ is a subtree of our graph.
If this is not the case,
we introduce a new connected component that contains only $B_{4k+3}$.
This permits to identify a color by its weight.
Indeed, in any coloring $c = (S_i)_{i \in \intv{0,\ell}}$, where $\ell \geq 4k+3$, of weight at most $M$, we have that for each $i \in \intv{0,4k+3}$,
$S_i$ is the only color such that $w(S_i) \in \{w_i^j \mid j \in \intv{0,n}\}$.
We denote, for each $\ell \in \intv{0,3}$, $R_\ell = S_{4k+\ell}$ to be the unique color of weight $W_\ell$.

We also define the \emph{auxiliary tree} $A_i^j$ for each $i \in \intv{0,4k-1}$ and each $j \in \intv{0,n}$, as defined in \cite{Julio}. This auxiliary tree is depicted in Figure~\ref{fig:aux}.

\begin{definition}
  For each $i \in \intv{0,4k-1}$ and each $j \in \intv{0,n}$, we define  the weighted rooted
tree $A_i^j$, called \emph{auxiliary tree}, as follows.
\begin{itemize}
\item We first introduce two vertices $u$ and $v$ such that $u$ is the root of $A_i^j$, $v$ is connected to $u$, $w(u) = W_0$, and $w(v) = w_i^j$.
\item for each $\ell \in \intv{0,i-2}$, we introduce a copy of $B_\ell$ and we connect the root of this copy to $v$.
\item for each $\ell \in \intv{0,4k-1} \sm \{i-1\}$, we introduce a copy of $B_\ell$ and we connect the root of this copy to $u$.
\end{itemize}
The vertex $v$ is called the \emph{subroot} of $A_i^j$. Note that $A_i^j$ consists of $2^{4k}$ nodes.
\end{definition}

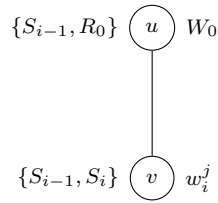
\begin{figure}[h!]
  \centering
  \begin{tikzpicture}
    \node[draw, circle, black, label={180:\scriptsize{$\{S_{i-1},R_0\}$}}, label={0:\scriptsize{$W_0$}}] (I1) at (0,2){\scriptsize{$u$}};
    \node[draw, circle, black, label={180:\scriptsize{$\{S_{i-1}, S_i\}$}}, label={0:\scriptsize{$w_i^j$}}] (I2) at (0,0){\scriptsize{$v$}};
    \draw (I2) -- (I1);
  \end{tikzpicture}
  \caption{
The auxiliary tree $A_i^j$, $i \in \intv{0,4k-1}$ and $j \in \intv{0,n}$.
The binomial trees are not depicted.
Next to each vertex, its weight and the set of colors this vertex can receive (see Lemma~\ref{lemma:aux}) are depicted.}
  \label{fig:aux}
\end{figure}

\begin{lemma}[Ara\'ujo \emph{et al.}\cite{Julio}]
\label{lemma:aux}
  Let $i \in \intv{0,4k-1}$, let $j \in \intv{0,n}$,
and let $(T, w)$ be any weighted forest having $B_{4k+3}$ and $A_i^j$ as  subtrees.
Let $u$ and $v$ be the root and the subroot of $A_i^j$, respectively.
For any coloring $c$ of $(T, w)$ with weight $w(c) ≤ M$, it holds that:
\begin{itemize}
\item either $v$ is colored $S_{i−1}$ and $u$ must be colored with the color  $R_0$,
\item or $v$ is colored $S_i$ (therefore, $w(S_i ) ≥ w_i^j$ ) and $u$ is  colored either with $S_{i−1}$ or with the color $R_0$.
\end{itemize}
\end{lemma}

We also need the $R_i$-\textsc{AND} gadget, $i \in \intv{0,1}$, depicted in Figure~\ref{fig:and}, and which is strongly inspired by a similar gadget presented in~\cite{Julio} (called \emph{clause tree}) corresponding to the logical `OR'.
\begin{definition}
  Let $i \in \intv{0,1}$. Given two vertices $I_1$, $I_2$, we define the \emph{$R_i$-\textsc{AND} gadget} between the \emph{input vertices} $I_1$ and $I_2$ as follows:
  \begin{itemize}
  \item We add four new vertices $v_1$, $v_2$, $v_3$, and $O$ and the edges $\{v_1,I_1\}$, $\{v_2,I_2\}$, $\{v_1,v_2\}$, $\{v_2,v_3\}$, and $\{v_3,O\}$.

  \item For each $j \in \intv{1,3}$ and each $\ell \in \intv{0,4k-1}$, we introduce a copy of $B_\ell$ and we connect its root to $v_j$.
  \item For each $\ell \in \intv{0,4k-1}$, we introduce a copy of $B_\ell$ and we connect its root to $O$.
  \item For each $j \in \intv{1,2}$ we introduce a copy of $B_{4k+1-i}$ and we connect its root to $v_j$.
  \item We introduce a copy of $B_{4k+i}$ and a copy of $B_{4k+2}$  and we connect their roots to $v_3$.
  \item We set $w(v_1) = W_2$, $w(v_2) = W_3$, $w(v_3) = W_3$, and $w(O) = W_1$.
  \end{itemize}
The vertex $O$ is called the \emph{output vertex} of the $R_i$-\textsc{AND} gadget.

We naturally extend the definition of the $R_i$-\textsc{AND} gadget to  $\ell$ input vertices with $\ell \geq 2$ by introducing $\ell -1$ $R_i$-\textsc{AND} gadgets in a sequential way.
\end{definition}

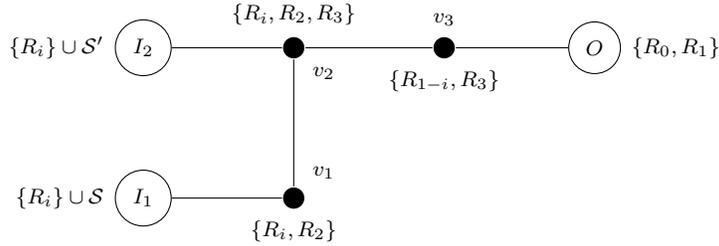
\begin{figure}[h!]
  \centering
  \begin{tikzpicture}
    \node[draw, circle, black, label={180:\scriptsize{$\{R_i\}\cup \Scal$}}] (I1) at (0,0){\scriptsize{$I_1$}};
    \node[draw, circle, black, label={180:\scriptsize{$\{R_i\}\cup \Scal'$}}] (I2) at (0,2){\scriptsize{$I_2$}};
    \node[place, label={90:\scriptsize{$\{R_i,R_2,R_3\}$}}, label={-45:\scriptsize{$v_2$}}] (A1) at (2,2){};
    \node[place, label={-90:\scriptsize{$\{R_i,R_2\}$}}, label={+45:\scriptsize{$v_1$}}] (A2) at (2,0){};
    \node[place, label={-90:\scriptsize{$\{R_{1-i},R_3\}$}}, label={90:\scriptsize{$v_3$}}] (A3) at (4,2){};
    \node[draw, circle, black, label={0:\scriptsize{$\{R_0,R_1\}$}}] (O) at (6,2){\scriptsize{$O$}};
    \draw (A1) -- (I2);
    \draw (O) -- (A3) -- (A1) -- (A2) -- (I1);
  \end{tikzpicture}
  \caption{The $R_i$-\textsc{AND} gadget, for some $i \in \intv{0,1}$, where $I_1$ and $I_2$ are the input vertices and $O$ is the output vertex, and where $\Scal$ and $\Scal'$ are subsets of $\{S_\ell \mid \ell \in \intv{0,4k-1}\} \cup \{W_0, W_1\}$. For each vertex, the associated set is the set of colors that the vertex can receive. Again, the binomial trees are not depicted.}
  \label{fig:and}
\end{figure}

\begin{lemma}
\label{lemma:and}
  Let $i \in \intv{0,1}$,
  let $I_1$ and $I_2$ be the two input vertices of an $R_i$-\textsc{AND} gadget, and let $O$ be its output vertex.
  If $I_1$ and $I_2$ are colored $R_i$, then $O$ must be colored $R_i$.
  Moreover, if either $I_1$ or $I_2$ is not colored $R_i$, then $O$ can be colored either $R_0$ or $R_1$.
\end{lemma}
\begin{proof}
  First, assume that $I_1$ and $I_2$ are colored $R_i$.
  This sequentially implies that $v_1$ must be colored $R_2$, $v_2$ must be colored $R_3$, $v_3$ must be colored $R_{1-i}$, and  $O$ must be colored $R_i$.
  Secondly, assume that $I_1$ is not colored $R_i$.
  This sequentially implies that $v_1$ can be colored $R_i$, $v_2$ can be colored $R_2$, $v_3$ can be colored $R_3$, and therefore $O$ can be colored either $R_0$ or $R_1$.
  Finally, assume that $I_2$ is not colored $R_i$.
  This sequentially implies that $v_2$ can be colored $R_i$, $v_3$ can be colored $R_3$, and so $O$ can be colored either $R_0$ or $R_1$.
\end{proof}

Finally, we define, for each  $i \in \intv{0,k-1}$ and $j \in \intv{0,n-1}$, the \emph{vertex tree} $T_i^j$, depicted in Figure~\ref{fig:tij}, which is also inspired by a similar construction given in~\cite{Julio}, called \emph{variable tree}. The main difference with respect to~\cite{Julio} is that, in our case, the color given to the root of a vertex tree codifies a binary value corresponding to picking or not a vertex in the solution, whereas the gadget of~\cite{Julio} codifies an integer corresponding to the assignment of a group of variables in the \emph{integral} version of $3$-\textsc{SAT} that they consider.

\begin{definition}
  For each $i \in \intv{0,k-1}$ and for each $j \in \intv{0,n-1}$, we define the \emph{vertex tree} $T_i^j$ to be the weighted rooted tree, representing the
  vertex $\beta^{-1}(j)$,  defined as follows.
  \begin{itemize}
  \item We introduce one copy of $A_{4i+1}^{j+1}$ and $A_{4i+3}^{n-j}$ and an $R_0$-\textsc{AND} gadget whose inputs are the two roots of $A_{4i+1}^{j+1}$ and $A_{4i+3}^{n-j}$. We call $u$ the output of the $R_0$-\textsc{AND} gadget and we set $u$ to be the root of $T_i^j$.
  \item We introduce one copy of $A_{4i+1}^j$, $A_{4i+1}^{j+1}$, $A_{4i+3}^{n-j}$, and $A_{4i+3}^{n-j-1}$,
    \begin{itemize}
    \item we connect $r(A_{4i+1}^j)$ to $r(A_{4i+3}^{n-j})$ and $r(A_{4i+1}^{j+1})$ to $r(A_{4i+3}^{n-j-1})$, and
    \item we connect $u$ to $r(A_{4i+1}^j)$ and to $r(A_{4i+3}^{n-j-1})$.
\end{itemize}
\end{itemize}
\end{definition}

\medskip

\begin{figure}[htb]
  \centering
  \begin{tikzpicture}
    \node[draw, circle, black, label={90:\scriptsize{$\{R_0,R_1\}$}}] (u) at (0,0){$u$};
    \node[place, label={-90:\scriptsize{$\{S_{4i},R_0\}$}}, label={90:\scriptsize{$A_{4i+1}^j$}}] (A1) at (-2,0){};
    \node[place, label={-90:\scriptsize{$\{S_{4i+2},R_0\}$}}, label={90:\scriptsize{$A_{4i+3}^{n-j}$}}] (A1) at (-4,0){};
    \node[place, label={-90:\scriptsize{$\{S_{4i},R_0\}$}}, label={90:\scriptsize{$A_{4i+1}^{j+1}$}}] (A1) at (4,0){};
    \node[place, label={-90:\scriptsize{$\{S_{4i+2},R_0\}$}}, label={90:\scriptsize{$A_{4i+3}^{n-1-j}$}}] (A1) at (2,0){};

    \node[draw, circle, black, label={-90:\scriptsize{$R_0$-\textsc{AND}}}] (and) at (0,-2){};
    \node[place, label={-90:\scriptsize{$\{S_{4i},R_0\}$}}, label={90:\scriptsize{$A_{4i+1}^{j+1}$}}] (A5) at (2,-2){};
    \node[place, label={-90:\scriptsize{$\{S_{4i+2},R_0\}$}}, label={90:\scriptsize{$A_{4i+3}^{n-j}$}}] (A6) at (-2,-2){};
    \draw (4,0) -- (u) -- (-4,0);
    \draw (A5) -- (and) -- (A6);
    \draw (u) -- (and);

\end{tikzpicture}
  \caption{
The vertex tree $T_i^j$, $i \in \intv{0,k-1}$ and $j \in \intv{0,n-1}$.
The vertices labeled $A_p^q$ are the roots of a copy of $A_p^q$.
For each vertex, the associated set is the set of colors that the vertex can receive.}
  \label{fig:tij}
\end{figure}
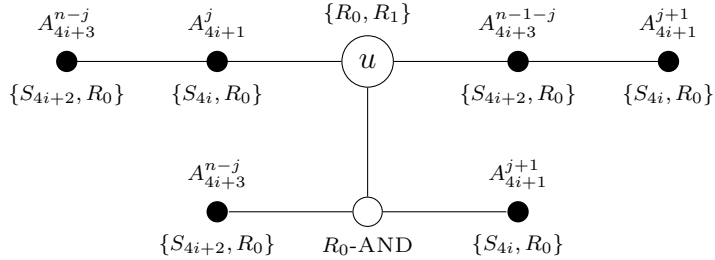

\medskip
\medskip

The usefulness of a vertex tree $T_i^j$ associated with a vertex $v$ corresponding to the integer $j$ is the following. The color of the root $u$ codifies whether vertex $v$ has been chosen in the solution or not. Namely, if $u$ gets color $R_0$ (resp. $R_1$), this means that vertex $v$ is (resp. is not) part of the solution. The following lemma formalizes this idea and guarantees that the choices are {\sl consistent}, in the sense that the choices made in all vertex trees corresponding to the same vertex are the same. It is also important to note that because of the definition of the weights $w_i^j$, each time we choose to color a root of a vertex tree with $R_0$, we have to ``pay'' $(n-1)\varepsilon$ in the total weight. Making $k$ such choices is forced by the properties of the gadgets and the value of $M$.

\newpage

\begin{lemma}
\label{lemma:tij}
Let $(T, w)$ be any weighted forest having $B_{4k+3}$ as a subtree and containing, for each $(i,j) \in \intv{0,k-1}\times \intv{0,n-1}$, $T_i^j$ as a subtree.
Let $c$ be a coloring of $(T,w)$ with  $w(c) \leq M$.
Then, there exist $(j_i)_{i \in \intv{0,k-1}} \in \intv{0,n-1}^k$ such that each root $u$ of each
subtree $T_i^j$, $(i,j) \in \intv{0,k-1}\times \intv{0,n-1}$, satisfies:

\begin{itemize}
\item if $j = j_{i}$ for some $i \in \intv{0,k-1}$, then the color of $u$ in $c$ must be $R_0$, and
\item otherwise, the color of $u$ in $c$ must be $R_1$.
\end{itemize}
\end{lemma}

\begin{proof}
By Lemma~\ref{lemma:bin} and since we assume that $w(c) \leq M$ and that $B_{4k+3}$ occurs in $(T,w)$ as a subtree, it follows that $c=(S_i)_{i \in \intv{0,\ell}}$ with $\ell \geq 4k+3$ and for each $i \in \intv{0,4k+3}$, $w(S_i) \in \{w_i^j\mid j \in \intv{0,n}\}$.

Let $i \in \intv{0,k-1}$.
Given  $j \in \intv{0,n}$,
as  $T_i^j$ or $T_i^{j-1}$ is a subgraph of $T$ (in fact, if $j\notin \{0,n-1\}$, both are), we know that there exist a copy of $A_{4i+1}^j$ with root $r_{4i+1}^j$ and a copy of $A_{4i+3}^{n-j}$ with root $r_{4i+3}^{n-j}$ such that $r_{4i+1}^j$ and $r_{4i+3}^{n-j}$ are adjacent.
This implies that, for each $j \in \intv{0,n}$,

\begin{eqnarray*}
&&c(r_{4i+1}^j) \not = R_0 \mbox{ or} ~~~~~~~~~~~~~~~~~~~~~~~~~~~~~~~~~~~~~~~~~~~~~~~~~~~~~~~~~~~~~~~~~~~~~~~~~~(1_j)\\
&&c(r_{4i+3}^{n-j}) \not = R_0.       ~~~~~~~~~~~~~~~~~~~~~~~~~~~~~~~~~~~~~~~~~~~~~~~~~~~~~~~~~~~~~~~~~~~~~~~~~~~~~(2_j)
\end{eqnarray*}
Note that, by Lemma~\ref{lemma:aux}, for each $j \in \intv{0,n}$, $(1_j)$ implies that $w(S_{4i+1}) \geq w_{4i+1}^j$ and $(2_j)$ implies that $w(S_{4i+3}) \geq w_{4i+3}^{n-j}$.
Therefore, one of the following cases necessarily occurs:
\begin{itemize}
\item $(1_n)$ is satisfied and so  $w(S_{4i+1}) \geq w_{4i+1}^n$,
\item $(2_0)$ is satisfied and so  $w(S_{4i+3}) \geq w_{4i+3}^{n-0}$, or
\item $(1_0)$ and $(2_n)$ are satisfied
and, since for each $j \in \intv{0,n}$ at least one of $(1_j)$ and $(2_j)$ holds, the
integer $j^* = \min \{j \mid 0 \leq j \leq n-1 \text{ and property $(2_{j+1})$ is satisfied}\}$ is well-defined. It follows that
 both $(1_{j^*})$ and $(2_{j^*+1})$ are satisfied, which implies that
$w(S_{4i+1}) \geq w_{4i+1}^{j^*}$ and $w(S_{4i+3}) \geq w_{4i+3}^{n-(j^*+1)}$.
\end{itemize}

In the first two cases, using that $w(S_{4i+1}) \geq w_{4i+1}^0$ and $w(S_{4i+3}) \geq w_{4i+3}^{0}$, we obtain $w(S_{4i+1}) + w(S_{4i+3}) \geq w_{4i+1}^0 + w_{4i+3}^0 + n\varepsilon$. In the third case, we obtain $w(S_{4i+1}) + w(S_{4i+3}) \geq (w_{4i+1}^0 + j^*\varepsilon) + (w_{4i+3}^0 + (n-(j^*+1))\varepsilon) = w_{4i+1}^0 + w_{4i+3}^0 + (n-1) \varepsilon$.

Thus, it always holds that $w(S_{4i+1}) + w(S_{4i+3}) \geq w_{4i+1}^0 + w_{4i+3}^0 + (n-1) \varepsilon$.

Therefore,
\begin{alignat*}{2}
w(c) & \geq  \sum_{i \in \intv{0,k-1}}(w(S_{4i})+w(S_{4i+1})+w(S_{4i+2})+w(S_{4i+3})) + \sum_{i \in \intv{0,3}} w(R_i) \\
     & \geq  \sum_{i \in \intv{0,k-1}}(w_{4i}^0+ w_{4i+1}^0 + w_{4i+2}^0+ w_{4i+3}^0 + (n-1)\varepsilon) +\sum_{i \in \intv{0,3}} W_i\\
     & =  M.
\end{alignat*}
By definition of $c$, we have  $w(c) = M$,  for each $i \in \intv{0,3}$, $w(R_i) = W_i$, and for each $i \in \intv{0,k-1}$,
$w(S_{4i}) = w_{4i}^0$, $w(S_{4i+2}) = w_{4i+2}^0$, and $w(S_{4i+1}) + w(S_{4i+3}) = w_{4i+1}^0 + w_{4i+3}^0 + (n-1)\varepsilon$.
Moreover, for each $4k+3 < i \leq \ell$, $w(S_i) = 0$.

\newpage

Let us fix $i^* \in \intv{0,k-1}$.
The equation $w(S_{4i^*+1}) + w(S_{4i^*+3}) = w_{4i^*+1}^0 + w_{4i^*+3}^0 + (n-1)\varepsilon$ implies the existence of $j^*\in \intv{0,n-1}$ such that
$w(S_{4i^*+1}) = w_{4i^*+1}^{j*}$ and $w(S_{4i^*+3}) = w_{4i^*+3}^{n-1-j*}$.
Thus, for each $j > j^*$, the root of any copy of $A_{4i^*+1}^{j}$ must be colored $R_0$ and
for each $j < j^*$, the root of any copy of $A_{4i^*+3}^{n-1-j}$ must be colored $R_0$.
This implies that for each $j \in \intv{0,n-1}\sm \{j^*\}$, the root of $T_{i^*}^{j}$ must be colored $R_1$.
Moreover, as in $T_{i^*}^{j^*}$ the roots of the copy of  $A_{4i^*+1}^{j^*+1}$ and the copy of $A_{4i^*+3}^{n-j^*}$ must be colored $R_0$ (otherwise, $w(S_{4i^*+1}) \geq w_{4i^*+1}^{j*+1}> w_{4i^*+1}^{j*}$ or $w(S_{4i^*+3}) \geq w_{4i^*+3}^{n-j*} > w_{4i^*+3}^{n-1-j*}$),  the $R_0$-\textsc{AND} gadget ensures that the root of $T_{i^*}^{j^*}$ is colored $R_0$.
\end{proof}

\section{$\W[1]$-hardness reduction}
\label{sec:W1}

In this section we present a parameterized reduction from
\textsc{Independent Set}
to \textsc{Weighted Coloring} on forests.


%
%
%
%

%

\begin{theorem}
  \label{th:W1}
Given a weighted forest $(G,w)$, the problem of computing  $\sigma(G,w)$ is $\W[1]$-hard when parameterized by the size of a largest connected component of $G$.
\end{theorem}



\begin{proof}
We reduce from \textsc{Independent Set} parameterized by the size of the solution, which is well-known to be $\W[1]$-hard (see~\cite{DF13, CyganFKLMPPS15}).
Let $(G,k)$ be an instance of \textsc{Independent Set}, and let $n = |V(G)|$. Recall that $M = k(n-1)\varepsilon + \sum_{i\in \intv{0,4k+3}}\frac{1}{2^i}$ where $\varepsilon$ is any real number satisfying $0 < \varepsilon < \frac{1}{n k 2^{4k+3}}$, which implies that $M < 2$. Let $\beta: V(G) \to \intv{0,n-1}$ be a bijection.
For each edge $\{v_1,v_2\}$ in $E(G)$ and each $i_1, i_2$ in $\intv{0,k-1}$, we define the weighted rooted tree $H_{\{v_1,v_2\},i_1,i_2}$ as follows.
\begin{itemize}
\item We introduce a copy of $T_{i_1}^{\beta(v_1)}$ and a copy of $T_{i_2}^{\beta(v_2)}$, and call the roots $r_1$ and $r_2$, respectively.
\item We introduce an $R_0$-\textsc{AND} gadget where the input vertices are $r_1$ and $r_2$ and the output is a new vertex $r$.
\item We introduce a copy of $B_{4k}$ and we connect its root to $r$.
\item We set $r$ to be the root of $H_{\{v_1,v_2\},i_1,i_2}$.
\end{itemize}

Note that, by construction, the root $r$ has to be colored $R_1$. We also define, for each vertex $v$ in $V(G)$  and each $i_1, i_2$ in $\intv{0,k-1}$ with $i_1 \neq i_2$, the weighted rooted tree $H_{v,i_1,i_2}$ to be the tree $H_{\{v_1,v_2\},i_1,i_2}$ defined above with $v_1 = v_2 = v$.

We define $(G',w)$ as the disjoint union of the weighted tree $B_{4k+3}$, of each weighted tree of $\{H_{e,i_1,i_2} \mid e \in E(G), i_1,i_2 \in \intv{0,k-1}\}$, of each weighted tree of $\{H_{v,i_1,i_2} \mid v \in V(G), i_1,i_2 \in \intv{0,k-1}, i_1 \neq i_2\}$, and
of each weighted tree of $\{T_i^j \mid i \in \intv{0,k-1}, j \in \intv{0,n-1}\}$.
Note that the size of each connected component of $G'$ is bounded by a function depending only on $k$. Indeed, the size of any connected component is bounded by the size of those of type $H_{e,i_1,i_2}$, which can be easily checked to be at most $2 \cdot (6 \cdot 2^{4k}+4)+4+2^{4k} = 13 \cdot 2^{4k}+12$. Note  that the construction of $(G',w)$ can be performed in time $f(k) \cdot n^{O(1)}$, as required.

\newpage

The idea of the construction is that the trees $H_{\{v_1,v_2\},i_1,i_2}$ defined above guarantee that, for each edge $\{v_1,v_2\}$ of $G$, at most one of $v_1$ and $v_2$ belongs to the independent set. More formally, as the root $r$ of such tree has to be colored $R_1$, by the $R_0$-\textsc{AND} gadget at least one of $r_1$ and $r_2$ has to be colored $R_1$, which translates to the fact that at least one of $v_1$ and $v_2$ does {\sl not} belong to the independent set. Similarly, the trees $H_{v,i_1,i_2}$ guarantee that the same vertex is not picked more than once in the solution.

%
%



More formally, we now prove that there exists a solution of \textsc{Independent Set} on $(G,k)$ if and only if $\sigma(G',w) \leq M$.


Assume first that $Z$ is a solution of  \textsc{Independent Set} on $(G,k)$.
We may assume that $Z$ is of size exactly $k$.
Let $\delta: Z \to \intv{0,k-1}$  be a bijection.
For each $i \in \intv{0,k-1}$, we define $v_i = \delta^{-1}(i)$.
We are going to define $c = (S_i)_{i \in \intv{0,4k+3}}$ such that for each $i \in \intv{0,4k+3}$, $w(S_i) \in \{w_i^j\mid j \in \intv{0,n}\}$.
By Lemma~\ref{lemma:bin}, we can (and we must) color every tree $B_i$ in that way, for each $i \in \intv{0,4k+3}$.
Then for each $j \leq \beta(v_i)$ and each $j' \geq \beta(v_i)$, we set the color of the subroot of each $A_{4i+1}^j$ and each $A_{4i+3}^{m-j'-1}$ to be
to be color $S_{4i+1}$ and $S_{4i+3}$,  respectively, and their root to be colored $S_{4i}$ and $S_{4i+2}$, respectively.
For each $j > \beta(v_i)$ and each $j' < \beta(v_i)$, we set the color of the roots of  each $A_{4i+1}^j$ and each $A_{4i+3}^{m-j'-1}$ to be $R_0$ and the color of their subroots to be $S_{4i+1}$ and $S_{4i+3}$, respectively.
This coloring is possible by Lemma~\ref{lemma:aux}.
Note also that for each $i \in \intv{0,k-1}$, if $j_i = \beta(v_i)$, then we have
$w(S_{4i}) = w_i^0$,
$w(S_{4i+1}) = w_i^{j_i}$,
$w(S_{4i+2}) = w_i^0$, and
$w(S_{4i+3}) = w_i^{m-{j_i}-1}$.
We set the color of the root of each $T_i^j$ such that $j = \beta(\delta^{-1}(i))$ to $R_0$,
and we set the color of the root of each $T_i^j$ such that $j \not = \beta(\delta^{-1}(i))$ to $R_1$.
The colors of the other vertices are forced by the $R_0$-\textsc{AND} gadgets.

As $Z$ is an independent set,  for each edge $\{v_1,v_2\}$ of $G$, at least one of the extremities, say $v_1$, is not in $Z$.
Thus, for each $i_1, i_2$ in $\intv{0,k-1}$, the root of $T_{i_1}^{\beta(v_i)}$ is colored $R_1$ and therefore the root of $H_{\{v_1,v_2\},i_1,i_2}$ can be colored $R_1$, which is the only color available for this vertex.
As in this coloring, for each $\ell \in \intv{0,3}$, $w(R_{\ell}) = W_\ell$, we obtain that $\sigma(G',w) \leq M$.


Conversely, assume that there is an integer $\ell$ and a coloring $c=(S_i)_{i \in \intv{0,\ell}}$ of $G'$ such that $w(c) \leq M$.
As there is no weight below $W_3$, from Lemma~\ref{lemma:bin} it follows that $\ell = 4k+3$ and for each $i \in \intv{0,4k+3}$, $w(S_i) \in \{w_i^j\mid j \in \intv{0,n}\}$.
By Lemma~\ref{lemma:tij}, for each $i \in \intv{0,k-1}$, there exists an index  $j_i$ such that the root of each $T_i^{j_i}$ is colored $R_0$.
Let us define $Z = \{\beta^{-1}(j_i)\mid i \in \intv{0,k-1}\}$.
Given $i_1$ and $i_2$ in $\intv{0,k-1}$, we claim that there is no edge in $G$ between $\beta^{-1}(j_{i_1})$ and $\beta^{-1}(j_{i_2})$.
Indeed, if the root of $T_{i_1}^{j_{i_1}}$ and the root of $T_{i_2}^{j_{i_2}}$ are colored $R_0$, then the root of
$H_{\{\beta^{-1}(j_{i_1}),\beta^{-1}(j_{i_2})\},i_1,i_2}$ should also be colored $R_0$ because of the $R_0$-\textsc{AND} gadget, but this is not possible because of the tree $B_{4k}$ that is connected to it. A similar argument shows that, because of the trees $H_{v,i_1,i_2}$, for any $i_1,i_2$ in $\intv{0,k-1}$ with $i_1 \neq i_2$, it holds that $\beta^{-1}(j_{i_1}) \neq \beta^{-1}(j_{i_2})$, that is, the same vertex does not occur  more than once in $Z$. This implies that $Z$ is an independent set in $G$ of size exactly $k$, concluding the proof.
\end{proof}

\newpage

\section{$\W[2]$-hardness reduction}
\label{sec:W2}

In this section we present a reduction from
\textsc{Dominating Set}
to \textsc{Weighted Coloring} on forests when the number of colors is prescribed. The reduction is similar to the one presented in Theorem~\ref{th:W1}, but it is somehow simpler and uses the $R_1$-\textsc{AND} gadget instead of the $R_0$-\textsc{AND} gadget.




%
%
%

%
%
%

\begin{theorem}\label{th:W2}
Given a weighted forest $(G,w)$ and a positive integer $r$, the problem of computing $\sigma(G,w;r)$ is $\W[2]$-hard when parameterized by $r$.
\end{theorem}


\begin{proof}
We reduce from \textsc{Dominating Set} parameterized by the size of the solution, which is well-known to be $\W[2]$-hard (see~\cite{DF13, CyganFKLMPPS15}).
Let $(G,k)$ be an instance of \textsc{Dominating Set}, and let $n = |V(G)|$.
Recall again that $M = k(n-1)\varepsilon + \sum_{i\in \intv{0,4k+3}}\frac{1}{2^i}$ where $\varepsilon$ is any real number satisfying $0 < \varepsilon < \frac{1}{n k 2^{4k+3}}$, which implies that $M < 2$. Let $\beta: V(G) \to \intv{0,n-1}$ be a bijection. 
For each vertex $v \in V(G)$, we define the weighted rooted tree $H_v$ as follows.
\begin{itemize}
\item For each $i \in \intv{0,k-1}$ and each $j \in \beta(N[v])$, we introduce a copy of $T_i^j$  and call its root $r_i^j$.
\item We introduce an $R_1$-\textsc{AND} gadget where the input vertices are the vertices of $\{r_i^j \mid i \in \intv{0,k-1}, j \in \beta(N[v])\}$, and let $r$ be the output.
\item We introduce a copy of $B_{4k+1}$ and we connect its root to $r$.
\item We set $r$ to be the root of $H_v$.
\end{itemize}

We then define $(G',w)$ as the disjoint union of the weighted tree $B_{4k+3}$, of each weighted tree of
$\{H_v \mid v \in V(G)\}$, and
of each weighted tree of $\{T_i^j \mid i \in \intv{0,k-1}, j \in \intv{0,n-1}\}$. Finally, we set $r = 4k+4$. Note that $r$ depends only on $k$ and that the construction of $(G',w)$ can be performed in time $f(k) \cdot n^{O(1)}$, as required.

%

The idea of this construction is to guarantee that a dominating set in $G$ must contain, for each $v \in V(G)$, at least one vertex in $N[v]$. In the tree $H_v$, this is captured by forbidding its root $r$ to be colored $R_1$, which by the $R_1$-\textsc{AND} gadget implies that at least one of the roots of the trees $T_i^j$ must be colored $R_0$, meaning that at least one vertex in $N[v]$ belongs to the solution.

Formally, we now prove that there exists a solution of \textsc{Dominating Set} on $(G,k)$ if and only if $\sigma(G',w;r) \leq M$.


First assume that $Z$ is a solution of  \textsc{Dominating Set} on $(G,k)$.
We may assume that $Z$ is of size exactly $k$.
Let $\delta: Z \to \intv{0,k-1}$  be a bijection.
For each $i \in \intv{0,k-1}$, we define $v_i = \delta^{-1}(i)$.
We are going to define $c = (S_i)_{i \in \intv{0,4k+3}}$ such that for each $i \in \intv{0,4k+3}$, $w(S_i) \in \{w_i^j\mid j \in \intv{0,n}\}$, in the same way we did for Theorem~\ref{th:W1}.
By Lemma~\ref{lemma:bin}, we can (and we must) color every tree $B_i$ in that way, for $i \in \intv{0,4k+3}$.
Then for each $j \leq \beta(v_i)$ and each $j' \geq \beta(v_i)$, we set the color of the subroot of each $A_{4i+1}^j$ and each $A_{4i+3}^{m-j'-1}$ to be
to be color $S_{4i+1}$ and $S_{4i+3}$,  respectively, and their root to be colored $S_{4i}$ and $S_{4i+2}$, respectively.
For each $j > \beta(v_i)$ and each $j' < \beta(v_i)$, we set the color of the roots of each $A_{4i+1}^j$ and each $A_{4i+3}^{m-j'-1}$ to be $R_0$ and the color of their subroot to be $S_{4i+1}$ and $S_{4i+3}$, respectively.
Again, this coloring is possible by Lemma~\ref{lemma:aux}.
Note also that for each $i \in \intv{0,k-1}$, if $j_i = \beta(v_i)$, then we have
$w(S_{4i}) = w_i^0$,
$w(S_{4i+1}) = w_i^{j_i}$,
$w(S_{4i+2}) = w_i^0$, and
$w(S_{4i+3}) = w_i^{m-{j_i}-1}$.
We set the color of the root of each $T_i^j$ such that $j = \beta(\delta^{-1}(i))$ to $R_0$,
and we set the color of the root of each $T_i^j$ such that $j \not = \beta(\delta^{-1}(i))$ to $R_1$.
The colors of the other vertices are forced by the $R_1$-\textsc{AND} gadgets.

As $Z$ is a dominating set of $G$,  for each $v \in V(G)$, at least one of the vertices $r_i^j$, $i \in \intv{0,k-1}$, $j \in \beta(N[v])$, of $H_v$ is colored $R_0$.
So we can affect the color $R_0$ to the root of $H_v$, which is, by construction, the only available color for this vertex.
As in this coloring, for each $\ell \in \intv{0,3}$, $w(R_{\ell}) = W_\ell$, we obtain that $\sigma(G',w;r) \leq M$. 


Conversely, assume that there is an integer $\ell$ and a coloring $c=(S_i)_{i \in \intv{0,\ell}}$ of $G'$ that $\sigma(G,w;\ell +1) \leq M$.
As there is no weight below $W_3$, from Lemma~\ref{lemma:bin} it follows that $\ell = 4k+3$ and for each $i \in \intv{0,4k+3}$, $w(S_i) \in \{w_i^j\mid j \in \intv{0,n}\}$.
By Lemma~\ref{lemma:tij},  for each $i \in \intv{0,k-1}$, there exists an index  $j_i$ such that the root of each $T_i^{j_i}$ is colored $R_0$.
Let us define $Z = \{\beta^{-1}(j_i)\mid i \in \intv{0,k-1}\}$, where the same vertex may have been chosen for different indices in $\intv{0,k-1}$.
Let $v$ be a vertex of $G$.
As, by construction, the root of $H_v$ can only receive the color $R_0$ in any coloring of weight at most $M$, this implies that at least one vertex $r_{i^*}^{j^*}$, $i^* \in \intv{0,k-1}$, $j^* \in \beta(N[v])$, of $H_v$ is colored $R_0$.
This implies, by Lemma~\ref{lemma:tij}, that $\beta^{-1}(j^*) \in Z$.
Moreover, $\beta^{-1}(j^*) \in N[v]$.
It follows that $Z$ is a dominating set in $G$ of size at most $k$.
\end{proof}

\medskip
\medskip

Note that the proof of Theorem~\ref{th:W2} shows that, if $(G,k)$ in an instance of \textsc{Dominating Set}, then the number of colors of the constructed instance satisfies $r = 4k+4 = O(k)$. Note also that it is easy to strengthen the lower bound given by Theorem~\ref{th:W2} to apply to {\sl trees} instead of forests. Indeed,  we can just add a new vertex $v$, attach it to every connected component of the forest $G'$ built in the reduction, and give to $v$ a weight that does not conflict with any of the weights of its neighbors. By possibly using a new color containing only $v$, it still holds that $r = O(k)$.

The above paragraph together with the fact that, assuming the ETH, \textsc{Dominating Set} parameterized by the size of the solution cannot be solved in time $f(k) \cdot n^{o(k)}$ for any computable function $f$~\cite{ChenHKX06} imply the following corollary.

\begin{corollary}\label{cor:optimal} Assuming the ETH, there is no algorithm that, given a weighted tree $(G,w)$ and a positive integer $r$, computes $\sigma(G,w;r)$ in time $f(r)\cdot n^{o(r)}$ for any computable function $f$.
\end{corollary}

In particular, Corollary~\ref{cor:optimal} implies that on forests, and more generally on graphs of bounded treewidth, the running time stated in Equation~(\ref{eq:algorithm}), which in this case is equal to $n^{O(r)}$, is asymptotically optimal under the ETH.

%

%

\newpage

{\small
\bibliographystyle{endm}
\bibliography{Biblio}
}
\end{document}

%% file: Package.tex
\usepackage[english]{babel}
\usepackage[utf8x]{inputenc}

\usepackage{amssymb,amstext}
\usepackage{cite,scrtime}
\usepackage{graphics}
\usepackage{amsmath,dsfont}
\usepackage{mathrsfs}
\usepackage{fancyhdr}
\usepackage{verbatim}
\usepackage{boxedminipage}
\usepackage{colortbl}
\usepackage{algorithm}
\usepackage[noend]{algorithmic}
\usepackage{tikz}

\usepackage{stmaryrd} 

\usepackage[T1]{fontenc}

\usepackage{graphicx}

\usepackage{hyperref}

%% file: Command.tex
\newcommand{\Acal}{\mathcal{A}}

\newcommand{\Scal}{\mathcal{S}}

\newcommand{\Rbb}{\mathbb{R}}

\definecolor{gray1}{gray}{0.775}

\newcommand{\sm}{\setminus}

\newcommand{\intv}[1]{\left[ #1 \right]}

\newcommand{\es}{\emptyset}

\newcommand{\fpt}{{\sf FPT}\xspace}
\newcommand{\W}{{\sf W}\xspace}
\newcommand{\np}{{\sf NP}\xspace}
\newcommand{\poly}{{\sf P}\xspace}
\newcommand{\xff}{\chi_{\sf FF}\xspace}

\renewenvironment{proof}{\noindent \textsc{Proof:}}{\hfill$\square$\medskip}

\usetikzlibrary{shapes,snakes}

\tikzstyle{place}=[
  circle,
  fill=black,
  thick,
  inner sep=0.1cm,
  minimum size=2mm
]